\documentclass[conference,a4paper]{APSIPA2020}
\usepackage{multirow}
\usepackage[dvips]{graphicx}
\usepackage{amsmath}
\usepackage[psamsfonts]{amssymb}
\usepackage{amsxtra}
\usepackage{threeparttable}
\usepackage{amsmath,amssymb,amsfonts}
\usepackage{cuted}
\usepackage{lipsum}
\usepackage{algorithmic}
\usepackage{graphicx}
\usepackage{textcomp}
\usepackage{xcolor}
\usepackage{graphicx}
\usepackage{amssymb}
\usepackage{mathtools}
\usepackage{cite}
\usepackage{algorithm}
\usepackage{algorithmic}
\usepackage{bm}
\def\tr{{\text{tr}}}
\def\BibTeX{{\rm B\kern-.05em{\sc i\kern-.025em b}\kern-.08em
    T\kern-.1667em\lower.7ex\hbox{E}\kern-.125emX}}

\begin{document}

\title{Sampling Policy Design for Tracking Time-varying Graph Signals with Adaptive Budget Allocation}

\author{%
\authorblockN{%
Xuan Xie, Hui Feng, Bo Hu
}
\authorblockA{%
%\authorrefmark{1}
Research Center of Smart Networks and Systems, School of Information Science and Technology\\
Fudan University, Shanghai, China \\
E-mail: \{xxie15, hfeng, bohu\}@fudan.edu.cn }
%
%\authorblockA{%
%\authorrefmark{2}
%Northwestern Polytechnical University, Xi'an, China\\
%E-mail: ynzhang@nwpu.edu.cn  Tel/Fax: +86-29-XXXXXXXX}
%
}
\newcommand{\tabincell}[2]{\begin{tabular}{@{}#1@{}}#2\end{tabular}}
\newtheorem{theorem}{Theorem}
\newtheorem{lemma}{Lemma}
\newtheorem{definition}{Definition}
\maketitle
\thispagestyle{empty}

\begin{abstract}
There have been many works that focus on the sampling policy design for static graph signals (GS), but few for time-varying GS.
In this paper, we concentrate on how to select vertices to sample and how to allocate the sampling budget for a time-varying GS to reduce tracking error.
In the Kalman Filter (KF) framework, the problem of sampling policy design and budget allocation is formulated as an infinite horizon sequential decision process, in which the optimal sampling policy is obtained by Dynamic Programming (DP). 
Since the optimal policy is intractable, an approximate algorithm is proposed by truncating the infinite horizon to two stages. 
By introducing a new tool for analyzing the convexity or concavity of composite functions, we prove that the truncated problem is convex so that it can be solved by standard tools.
Finally, we demonstrate the performance of the proposed approach through numerical experiments.
\end{abstract}
\section{Introduction}
Time-varying graph signals (GS) are versatile for describing dynamics of signals in irregular domains, such as social, sensor and brain networks. 
There have been various works on stationariness, filtering and sampling of time-varying GS.
For example, stationary processes of GS and some corresponding applications are investigated in \cite{girault2015stationary,perraudin2017stationary}, frequency analysis and time-graph filter are proposed in \cite{loukas2016frequency,isufi2016separable,bohannon2019filtering}, time-varying GS reconstruction and sampling are introduced in \cite{mao2019time,wei2019optimal}.

For some large-scale networks, sampling theory for GS is essential since it is almost impractical to acquire the signals on all the nodes. 
Instead, the whole GS has to be estimated from the samples on a subset of nodes.
For example, the opinions of all the users in a network is usually estimated by sending questionnaires to part of the users, since it is unaffordable to obtain the opinion of everyone in a huge social network due to the limited time and manpower.  

When estimating GS from noisy observations, different sampling policies will result in different estimation performance. 
The design of sampling set aims at sampling vertices under the budget constraints to minimize the estimation error.
There are many works that focus on the sampling set design for static GS \cite{anis2016efficient,chen2015discrete,xie2017design}. 
Most existing sampling policy for time-varying GS minimize the tracking error myopically.  
In \cite{di2018adaptive}, sampling policies are designed for tracking bandlimited time-varying GS under least mean squares (LMS) and recursive least squares (RLS) framework.
Sampling strategy for tracking bandlimited GS by Kalman Filter (KF) is proposed in \cite{isufi2020observing}. 
In a related but not identical scenario, sensor selection is designed for target tracking in the network by KF \cite{shen2014sensor} and extended KF \cite{chepuri2014sparsity}.
However, \cite{di2018adaptive,isufi2020observing,shen2014sensor,chepuri2014sparsity} do not consider the effect of the present sampling policy to the future tracking error, which may be not optimal for the long-term tracking performance. 

The evolution of time-varying GS can depict both the change of signal like heat diffusion and the change of topology.
For example, the diffusion of GS on a random edge sampling (RES) graph \cite{7931690} which describes the case like link failures on the communication network,or street closures on the street network.
Considering that the evolution may be slow or abrupt, \emph{a good sampling policy design should not only focus on the instant tracking performance but also consider the long-term performance}.
Meanwhile, a reasonable allocation of sampling budget among time steps will be also beneficial for tracking.

In KF framework, we consider the problem of sampling policy design for tracking a time-varying GS over an infinite horizon with a given average budget. 
Different from \cite{di2018adaptive,isufi2020observing}, whose sampling sets are designed to minimize the instant tracking error, we consider the influence of current sampling policy to the future tracking performance and aim to minimize the tracking error for the long-term. 
Instead of given a fixed sampling budget for each time step, we also try to adaptively allocate the budget to minimize the tracking error.
The problem of sampling set design and budget allocation is formulated as an infinite horizon sequential decision process and solved by dynamic programming (DP). 
Furthermore, an approximate optimization problem is proposed to get a suboptimal solution since the optimal solution is computationally prohibitive. 
We also prove that the approximate optimization problem is convex by introducing a new tool to analyze the convexity of composite matrix valued functions so that it can be solved by standard tools.
Finally, several experiments validate that our approach has significantly improved the tracking performance than the state-of-art methods especially when the evolution of the GS is abrupt.

\section{Time-Varying Graph Signals}
Consider an $N$-vertex undirected connected graph $\mathcal{G}=(\mathcal{V},\mathcal{E},\mathbf{W})$, where $\mathcal{V}$ is the vertex set, $\mathcal{E}$ is the edge set, and $\mathbf{W}$ is the weighted adjacency matrix. If there is an edge $e=(i,j)$ between vertex $i$ and $j$, then $W_{i,j}$ represents the weight of the edge; otherwise $W_{i,j}=0$. A time-varying graph signal $\mathbf{f}_t \in \mathbb{R}^N$ at the moment $t$ has the element $(\mathbf{f}_t)_i$ representing the signal value on the $i$-th vertex in $\mathcal{V}$.

The graph Laplacian is defined as $\mathbf{L}=\mathbf{D}-\mathbf{W}$, where the weighted degree matrix $\mathbf{D}=\text{diag}(\mathbf{1}^\text{T}\mathbf{W})$ and $\mathbf{1}$ is a vector with all ones. Since the Laplacian matrix is real symmetric, it has a complete eigenbasis and the spectral decomposition 
$\mathbf{L}=\mathbf{V}\mathbf{\Lambda}\mathbf{V}^\text{T}$,
where the eigenvectors  $\{\mathbf{u}_k\}_{0\leq k\leq N-1}$ of $\mathbf{L}$ form the columns of $\mathbf{V}$, and $\mathbf{\Lambda}
\in \mathbb{C}^{N\times N}$ is a diagonal matrix of eigenvalues $0=\lambda_0\leq \lambda_1\leq \cdots\leq \lambda_{N-1}$ of $\mathbf{L}$. The Graph Fourier Transform (GFT) corresponds to the basis expansion of a signal. The eigenvectors of $\mathbf{L}$ are regarded as the graph Fourier bases and the eigenvalues are regarded as frequencies \cite{sandryhaila2014discrete}. The expansion coefficients of a static graph signal $\mathbf{f}$ in terms of eigenvectors are defined as $\hat{\mathbf{f}}$, so that a GFT pair can be expressed as $\mathbf{f}=\mathbf{V}\hat{\mathbf{f}}$ and $\hat{\mathbf{f}}=\mathbf{V}^\text{T}\mathbf{f}$.

%A GS can be estimated from partial observations only if we have known some prior knowledge of it.
In this paper, we assume that the GS is a stochastic signal and the stochastic prior is usually given in frequency domain \cite{8047995,perraudin2017stationary}, such that $\hat{\mathbf{f}}$ is drawn from the following distribution
\begin{equation}
\label{stochastic}
p(\hat{\mathbf{f}}) \propto \text{exp}(-(\hat{\mathbf{f}} - \bm{\mu})^\text{T}\mathbf{\Sigma}^{-1}_{\hat{\mathbf{f}}}(\hat{\mathbf{f}} - \bm{\mu})),
\end{equation}
where $p(\cdot)$ denotes probability density function, $\bm{\mu}$ and $\mathbf{\Sigma}_{\hat{\mathbf{f}}}= \text{diag}(\sigma^2_1,\cdots,\sigma^2_N)$ are the mean and covariance matrix of $\hat{\mathbf{f}}$ respectively.

We assume that the time-varying GS follow a predefined evolution matrix $\mathbf{H}_t$, which can be used to depict the network dynamics in both GS and topology, for example, disease progression \cite{raj2012network}, opinion propagation \cite{wu2014opinionflow} and topology with random edge connections \cite{7931690}.
The evolution noise is introduced to fit the uncertainty. Specifically, we have
\begin{subequations}
\label{linear model}
\begin{align}
\label{evolution}
	\mathbf{f_t}  = &\mathbf{H}_t\mathbf{f_{t-1}} + \mathbf{v}_{t-1},\quad\mathbf{v_t}\sim\mathcal{N}(\mathbf{0},\mathbf{\Sigma}_{\mathbf{v}}),\\
\label{observation}
	\mathbf{y}_t  = &\mathbf{\Psi}_t ({\mathbf{f}_t}+\mathbf{w}_t),\quad\mathbf{w_t}\sim\mathcal{N}(\mathbf{0},\mathbf{\Sigma}_{\mathbf{w}}),
\end{align}
\end{subequations}
where $\mathbf{w_t}$ and $\mathbf{v_t}$ are uncorrelated $\textit{i.i.d.}$ Gaussian noise, $\mathbf{\Sigma}_{\mathbf{v}} = \sigma^2_{\mathbf{v}}\mathbf{I}$ and $\mathbf{\Sigma}_{\mathbf{w}} = \sigma^2_{\mathbf{w}}\mathbf{I}$.
Eq. (\ref{evolution}) represents the GS evolution model and (\ref{observation}) is the observation model.
The sampling operator $\mathbf{\Psi}_t:\mathbb{R}^N\mapsto \mathbb{R}^{M}$ is defined as 
\begin{eqnarray}
\label{sampling operator}
(\mathbf{\Psi}_t)_{i,j}=
        \left\{\begin{matrix}
        1, & \text{if } j\text{-th node is sampled at time }t;\\
        0, & \text{otherwise.}
        \end{matrix}\right.
\end{eqnarray}

%The evolution matrix $\mathbf{H}_t$ can be constant over time or time-varying.
%We consider the more general case with time-varying $\mathbf{H}_t$, which can be used to model the evolution pattern change of GS like random walk on the graph or topology change of GS like random edge sampling (RES) graph \cite{7931690}.

%In some case, the evolution matrix $\mathbf{H}_t$ can be diagonalized by $\mathbf{V}$.
%For example, a general form depicting network dynamics in many real world scenarios, such as disease progression \cite{raj2012network}, opinion propagation \cite{wu2014opinionflow} and image smoothing \cite{ma2016diffusion} can be formulated as the polynomial or the rational fraction of Laplacian matrix or normalized Laplacian matrix which can be diagonalized by $\mathbf{V}$.
%In addition, the translation of a signal on graph \cite{shuman2013emerging} to vertex $i$ can be formulated as $\mathbf{H}_t = \mathbf{V}\text{diag}(\mathbf{V}^\text{T}\bm{\delta}_i)\mathbf{V}^\text{T}$, where $\bm{\delta}_i$ is an $N$-dimension vector with 1 on the $i$th element and 0 on the others.

%Since in some case, the GS is always evolving in a narrow low frequency band, which means fewer parameters are needed to be estimated in the spectral domain, so we change the tracking to graph spectral domain.

%The KF can be applied to track the GS.
By applying GFT, we change the GS to spectral domain and rewrite (\ref{linear model}) as
\begin{subequations}
\label{linear model fre}
\begin{align}
\label{evolution ref}
\hat{\mathbf{f}}_t =& \mathbf{V}^\text{T}\mathbf{H}_t\mathbf{V}\hat{\mathbf{f}}_{t-1} + \mathbf{V}^\text{T}\mathbf{v}_t = \tilde{\mathbf{H}}_t\hat{\mathbf{f}}_{t-1} + \mathbf{V}^\text{T}\mathbf{v}_t,\\
\label{observation ref}
\mathbf{y}_t = &\mathbf{\Psi}_t (\mathbf{V}\hat{\mathbf{f}}_t+\mathbf{w}_t),
\end{align}
where $\tilde{\mathbf{H}}_t = \mathbf{V}^\text{T}\mathbf{H}_t\mathbf{V}$.
% is a diagonal matrix.

\end{subequations}
%Since in the diffusion kernels on graphs\cite{kondor2002diffusion}, In signal difussion such as disease progression \cite{raj2012network}, opinion propagation \cite{wu2014opinionflow} and image smoothing \cite{ma2016diffusion}, the signal of one node is usually affected by its neighbour nodes.  The state transition matrix $\mathcal{L}$ is defined as a polynomial of Laplacian matrix. 
%For our theoretical derivations, we define $\mathcal{L}$ as Perron matrix \cite{olfati2004consensus} $\mathbf{I}-\epsilon\mathbf{L}$, where $\epsilon\in (0,1/d_{\text{max}})$ and $d_{\text{max}}$ is the maximal degree of $\mathcal{G}$. 

%where $\mathcal{S}_t=\left((\mathcal{S}_t)_1,\cdots,(\mathcal{S}_t)_{M}\right)$ denotes the sequence of sampled indices, $M$ denotes the sampling budget at each moment and $(\mathcal{S}_t)_i\in \{1,2,\cdots,N \}$. $\mathbf{A}_{\mathcal{S}_t}$ denotes the sub-matrix of $\mathbf{A}$ with rows and columns indexed by ${\mathcal{S}_t}$.

%For the linear dynamic system (\ref{linear model}), KF \cite{simon2006optimal} is used for signal tracking.
The KF can be applied for tracking the GS described by (\ref{linear model fre}), which consists of the following equations for each time step $t = 1,2,\cdots$:
\begin{eqnarray}
\label{prior estimation}
\hat{\mathbf{f}}_{t}^- &=& \tilde{\mathbf{H}}_t\hat{\mathbf{f}}_{t-1}^+,\\
\label{a prior cov}
\mathbf{P}_t^- &=& \tilde{\mathbf{H}}_t\mathbf{P}_{t-1}^+\tilde{\mathbf{H}}_t + \mathbf{\Sigma}_{\mathbf{v}},\\
\label{KF gain}
\mathbf{K}_t  &=& \mathbf{P}_t^-\mathbf{V}^\text{T}{\mathbf{\Psi}}_t^\text{T}\left({\mathbf{\Psi}}_t(\mathbf{V}\mathbf{P}_t^-\mathbf{V}^\text{T} + \mathbf{\Sigma}_{\mathbf{w}}){\mathbf{\Psi}}_t^\text{T} \right)^{-1},\\
\label{posterior estimation}
\hat{\mathbf{f}}_{t}^+ &=& \hat{\mathbf{f}}_{t}^- + \mathbf{K}_t(\mathbf{y}_t - \mathbf{\Psi}_t\mathbf{V} \hat{\mathbf{f}}_{t}^-),\\
\label{a posteriori cov}
\mathbf{P}_t^+ &=& \left((\mathbf{P}_t^-)^{-1} + \mathbf{V}^\text{T}\mathbf{\Psi}_t^\text{T}\mathbf{\Psi}_t\mathbf{\Sigma}_{\mathbf{w}}^{-1}{\mathbf{\Psi}}_t^\text{T}\mathbf{\Psi}_t\mathbf{V}\right)^{-1},
\end{eqnarray}
where $\hat{\mathbf{f}}_t^-$, $\mathbf{P}_t^-$, $\mathbf{K}_t $, $\hat{\mathbf{f}}_{t}^+$, $\mathbf{P}_t^+$ denote the prior estimation, prior covariance, KF gain, posterior estimation and posterior covariance respectively. 
The initialization states are $\hat{\mathbf{f}}_{0}^+ = \bm{\mu}$ and $\mathbf{P}_0^+=\mathbf{\Sigma}_{\hat{\mathbf{f}}}$ based on (\ref{stochastic}). 

\section{Sampling Policy Design}
\label{sce:sampling}
\subsection{Sampling as an Infinite Horizon Decision Process}
\label{subsec:DP}
In the KF framework, the mean squared error (MSE) of GS estimation in time step $t$ can be calculated by $\tr\left(\mathbf{P}_t^+\right)$.
In order to get an optimal tracking performance for an infinite horizon, we design a sequence of sampling operators $\{\mathbf{\Psi}_1,\mathbf{\Psi}_2,\cdots\}$ that minimize the accumulated tracking error $\sum_{t=1}^{\infty}\gamma^t\tr\left(\mathbf{P}_t^+\right)$ under the sampling budget constraints, where $\gamma\in(0,1)$ is the discount factor.
%\begin{eqnarray}
%\label{minJ}
%J = \underset{\mathbf{\Psi}_1,\mathbf{\Psi}_2,\cdots}{\text{min}} \sum_{k=1}^T\text{tr}\left(\mathbf{P}_t^+\right),
%\end{eqnarray}
%where $\text{tr}(\mathbf{P}_t^+)$ denotes the trace of the posterior covariance (\ref{a posteriori cov}) of KF in Section \ref{sec:KF}, which characterizes the tracking performance of time $k$. 
%which can be expressed by  
%\begin{align}
%\label{maxJ}
%\underset{\mathbf{\Psi}_1,\mathbf{\Psi}_2,\cdots}{\text{min}}\,J = &\sum_{k=1}^{\infty}\gamma ^t\tr[(\mathbf{P}_t^-)^{-1} +
%& \mathbf{V}^T\mathbf{\Psi}_t^T\mathbf{\Psi}_t\mathbf{\Sigma}_{\mathbf{w}}^{-1}{\mathbf{\Psi}}_t^T\mathbf{\Psi}_t\mathbf{V}]^{-1},
%\end{align}
%where $\gamma\in (0,1)$.
%\label{maxJ_sep}
%=\sum_{k=1}^T\text{tr}\left[(\mathbf{P}_t^-)^{-1} + \mathbf{\Psi}_t^T\mathbf{R}_t^{-1}\mathbf{\Psi}_t\right].
%In order to get an optimal tracking performance for an infinite horizon, we design a sequence of sampling operator $\{\mathbf{\Psi}_1,\mathbf{\Psi}_2,\cdots\}$ to maximize the trace of total certainty.
%\begin{eqnarray}
%\label{maxJ}
%&\underset{\mathbf{\Psi}_1,\mathbf{\Psi}_2,\cdots}{\text{max}}\quad & J = \sum_{k=1}^{\infty}\text{tr}\left[(\mathbf{P}_t^+)^{-1}\right],\\
%& \quad \text{s.t.} & \text{tr}(\mathbf{\Psi}_t^T\mathbf{\Psi}_t) = M \quad k=1,2,\cdots.\nonumber
%\end{eqnarray}
The effect of $\mathbf{\Psi}_t$ to the accumulated tracking error is cascading according to (\ref{a posteriori cov}) and (\ref{a prior cov}), so the sampling operator design at present time must balance the present tracking error and the future tracking error. 
%\emph{The infinite horizon sequential decision process is the standard framework to deal with this trade-off}.

Denote $(\mathbf{P}_t^-)^{-1}$ and $\mathbf{D}_t$ as the state and action of the system at time step $t$ respectively, where $\mathbf{D}_t \triangleq \mathbf{\Psi}_t^\text{T}\mathbf{\Psi}_t = \text{diag}(d_{1,t},\cdots,d_{N,t})$ is a diagonal matrix with the $i$-th diagonal element equalling 1 if the $i$-th vertex is sampled, and 0 elsewhere. 
The decision process can be formulated as a 5-tuple $(\mathcal{S},\mathcal{A},f_t,g_t,\gamma)$, where $\mathcal{S}$ is the state set of symmetric positive definite matrices $(\mathbf{P}_t^-)^{-1}$,
$\mathcal{A}$ is the action set of $N\times N$ matrices $\mathbf{D}_t$,
%Denote $\{(\mathbf{P}_1^-)^{-1},\cdots,(\mathbf{P}_T^-)^{-1}\}$ as a set of states in a infinite state space $\mathbf{S}_+^N$. And the state at each time summarizes the past tracking performance and signal evolution that is relevant for the future optimization; Denote $\{\mathbf{\Psi}_1,\cdots,\mathbf{\Psi}_T\}$ as a set of actions in a finite action space $A$, where $A$ is a set of $M\times N$ matrix with a 1 in different position of each line and zeros everywhere else.
$f_t$ is the law of the state transition with the form of
%\begin{eqnarray}
%\label{DP system}
$(\mathbf{P}_{t+1}^-)^{-1} = f_t((\mathbf{P}_t^-)^{-1},\mathbf{D}_t))$.
%\end{eqnarray}
According to (\ref{a posteriori cov}) and (\ref{a prior cov}) in KF, the state transition guided by $\mathbf{D}_t$ is
\begin{align*}
%\label{state transit}
(\mathbf{P}_{t+1}^-)^{-1} = &\left[\tilde{\mathbf{H}}_{t+1}\left((\mathbf{P}_t^-)^{-1} +  \sigma_{\mathbf{w}}^{-2}\mathbf{V}^\text{T}\mathbf{D}_t\mathbf{V}\right)^{-1}\tilde{\mathbf{H}}_{t+1} + \mathbf{\Sigma}_{\mathbf{v}}\right]^{-1}.
\end{align*}
$g_t$ is the immediate cost of action defined as the estimation error of instant estimation with the form of
%\begin{eqnarray}
$g_t((\mathbf{P}_t^-)^{-1},\mathbf{D}_t))$, which is affected by the present sampling policy $\mathbf{D}_t$,
%The trace of the information matrix $\text{tr}\left[(\mathbf{P}_t^+)^{-1}\right]$ in (\ref{maxJ}) is regarded as the immediate reward $g_t$. According to (\ref{a posteriori cov}), it has the form of
\begin{eqnarray}
\label{cost}
g_t((\mathbf{P}_t^-)^{-1},\mathbf{D}_t) = \tr\left(\mathbf{P}_t^+\right)= \tr\left((\mathbf{P}_t^-)^{-1} + \sigma_{\mathbf{w}}^{-2}\mathbf{V}^\text{T}\mathbf{D}_t\mathbf{V}\right)^{-1}.\nonumber
\end{eqnarray}
%\end{eqnarray}
%And $\gamma$ is the discount factor for future reward.  
%Then the optimal problem \ref{maxJ} is to choose a sequence of actions in order to maximize the total reward over a infinite horizon,
%\begin{eqnarray}
%\label{totalcost}
%& \underset{\mathbf{\Psi}_1,\mathbf{\Psi}_2,\cdots}{\text{max}}\quad & J = \sum_{k=1}^{\infty}\gamma^k g_t((\mathbf{P}_{k}^-)^{-1},\mathbf{\Psi}_t),\\
%& \quad \text{s.t.} & \text{tr}(\mathbf{\Psi}_t^T\mathbf{\Psi}_t) = M \quad k=1,2,\cdots.\nonumber
%\end{eqnarray}

%Since $\mathbf{Q}_t$ and $\mathbf{R}_t$ are diagonal, and $\mathcal{L}$ is a low-pass graph filter, $\mathcal{L}\mathbf{P}_{k-1}^+\mathcal{L}^T + \mathbf{Q}_{k-1}$ and $(\mathbf{P}_t^-)^{-1} + \mathbf{\Psi}_t^T\mathbf{\Psi}_t\mathbf{R}_t^{-1}\mathbf{\Psi}_t^T\mathbf{\Psi}_t$ have a natural tendency to be diagonally dominant. 
%And it has been proved that the Taylor series expansion of the inverse function converges for diagonally dominant matrices \cite{babu2019inverse}. 

Then the optimal problem for sampling policy design can be formulated as choosing a sequence of actions in order to minimize the total cost over an infinite horizon,
\begin{eqnarray}
\label{totalcost}
& \underset{\mathbf{D}_1,\mathbf{D}_2,\cdots}{\text{min}}\quad & J = \sum_{t=1}^{\infty}\gamma^t g_t((\mathbf{P}_t^-)^{-1},\mathbf{D}_t)\\
& \quad \text{s.t.} & 0\leq \text{tr}(\mathbf{D}_t) \leq M_t \quad t=1,2,\cdots,\nonumber\\
&\qquad & \lim_{T\to \infty}\frac{1}{T}\sum_{t=1}^\text{T}\tr(\mathbf{D}_t) = M,\nonumber\\
&\qquad & \mathbf{D}_t\subset \mathcal{A}\quad t=1,2,\cdots. \nonumber
\end{eqnarray}
The Bellman equation \cite{bertsekas2005dynamic} is used to compute the optimal action for the decision process sequentially,
%For $k = T,\cdots,1$, the cost-to-go function $J_t((\mathbf{P}_{k}^-)^{-1}))$ is related to $J_{k+1}((\mathbf{P}_{k+1}^-)^{-1}))$ through the recursion 
\begin{eqnarray}
\label{bellman}
J_t((\mathbf{P}_{t}^-)^{-1}) &=&  \underset{\mathbf{D}_{t}}{\text{min}}\left\{g_t((\mathbf{P}_t^-)^{-1},\mathbf{D}_t)\right.\nonumber\\
&&\left. + \gamma J_{t+1}(f_{t+1}((\mathbf{P}_t^-)^{-1},\mathbf{D}_t))\right\},
\end{eqnarray}
which means the optimal policy at $t$ is the one that minimizes the sum of immediate cost and future costs.
%The cost-to-go function for $k=T$ is given by
%\begin{eqnarray}
%J_T((\mathbf{P}_{T}^-)^{-1}) = \underset{\mathbf{\Psi}_{T}}{\text{max}} \left\{(g_T((\mathbf{P}_{T}^-)^{-1},\mathbf{\Psi}_T)\right\}.
%\end{eqnarray} 
%\vspace{-0.3cm}
\subsection{The Truncated Problem}
However, finding an optimal solution for (\ref{bellman}) is computational intractable. One reason is that the dimension of action space grows exponentially with the tracking time $t$. According to (\ref{totalcost}), the weight of immediate cost in the total cost decreases over time, which means the cost in the near future has a bigger impact on the total cost. So we truncate the infinite horizon future cost in (\ref{bellman}) to the length of one \cite{Farahmand2016TruncatedAD,4289630}, which means the policy of each time step only minimizes the sum of immediate cost and the cost of the next time. For each time step $t$, the future cost is truncated to the optimal cost of $t+1$ as
%\vspace{-0.15cm}
\begin{align*}
& J_{t+1}((\mathbf{P}_{t+1}^-)^{-1}) = \underset{\mathbf{D}_{t+1}}{\text{min}} g_{t+1}((\mathbf{P}_{t+1}^-)^{-1},\mathbf{D}_{t+1})\nonumber\\
& =\underset{\mathbf{D}_{t+1}}{\text{min}}\, \tr\left((\mathbf{P}_{t+1}^-)^{-1} + \sigma_{\mathbf{w}}^{-2}\mathbf{V}^\text{T}\mathbf{D}_{t+1}\mathbf{V}\right)^{-1}.
\end{align*}
%\vspace{-0.12cm}
Thus, we obtain a new one-step-look-ahead object function,
\begin{align}
\label{subopt}
& J_t((\mathbf{P}_t^-)^{-1})  = \underset{\mathbf{D}_t}{\text{min}}\,\left\{\tr\left((\mathbf{P}_t^-)^{-1} +\sigma_{\mathbf{w}}^{-2}\mathbf{V}^\text{T}\mathbf{D}_t\mathbf{V}\right)^{-1}\right.\nonumber\\
 & \left.+ \gamma \underset{\mathbf{D}_{t+1}}{\text{min}}\left\{\tr\left(f_t((\mathbf{P}_t^-)^{-1},\mathbf{D}_t)) + \sigma_{\mathbf{w}}^{-2}\mathbf{V}^\text{T}\mathbf{D}_{t+1}\mathbf{V}\right)^{-1}\right\}\right\}.
\end{align}
Also truncating allocation of the sampling budget to two time steps, the new optimization problem is as follow 
\begin{align}
\label{horizontwo}
 \underset{\tilde{\mathbf{D}}_{t},\tilde{\mathbf{D}}_{t+1}}{\text{min}}\, &\tr\left((\mathbf{P}_t^-)^{-1} +\sigma_{\mathbf{w}}^{-2}\mathbf{V}^\text{T}\tilde{\mathbf{D}}_t\mathbf{V}\right)^{-1}+\nonumber\\
\quad & \gamma \tr\left(f_t((\mathbf{P}_t^-)^{-1},\tilde{\mathbf{D}}_t)) + \sigma_{\mathbf{w}}^{-2}\mathbf{V}^\text{T}\tilde{\mathbf{D}}_{t+1}\mathbf{V}\right)^{-1} \\
 \quad \text{s.t.}\quad & 0\leq \text{tr}(\tilde{\mathbf{D}}_t) \leq M_t,\nonumber\\
\qquad & 0\leq \text{tr}(\tilde{\mathbf{D}}_{t+1}) \leq M_{t+1},\nonumber\\
\qquad & \tr(\tilde{\mathbf{D}}_t)+\tr(\tilde{\mathbf{D}}_{t+1}) = 2M,\nonumber\\
\qquad & \tilde{\mathbf{D}}_t,\tilde{\mathbf{D}}_{t+1}\subset \tilde{\mathcal{A}}, \nonumber
\end{align}
where $\tilde{\mathcal{A}}$ is a set of $N\times N$ matrices $\tilde{\mathbf{D}}_t$ with element $0\leq d\leq 1$ in the diagonal line and 0 everywhere else.
Compared with problem (\ref{totalcost}), we relax $d_{i,t}$ and $d_{i,t+1}$ to continuous values in $[0,1]$ since the optimization problem is an intractable combinatorial problem before relaxing.
%The convexity of optimization prblem (\ref{horizontwo}) will be analyzed in the Section \ref{section:analysis}.
%
%We can observe from the form of (\ref{Psi}) that our sampling policy is a trade-off between exploitation and exploration. It is not hard to find that if we just maximize the immediate reward (\ref{cost}) and ignores the future cost, the sampling policy is the same as the one obtained by maximizing the first term of (\ref{newDP}). This policy is an exploitation by sampling the $M$ nodes with the minimal observation noise variance, which means the policy will always select the nodes with large certainty. The second term in (\ref{Psi}) is obtained by considering the future reward, which not only involves the information of observation but also takes the prediction error and evolution of the signal into consideration, so it is a process of exploration is some sense. Thus, our sampling policy makes full use of the information for tracking by taking a trade-off between exploitation and exploration. 
The design of sampling policy for KF is described in Algorithm \ref{algorithm:DP}.
\begin{algorithm}[htb] 
\caption{Sampling policy design and GS tracking.} 
\label{algorithm:DP} 
\begin{algorithmic}[1] %这个1 表示每一行都显示数字
\STATE \textbf{Initialize} $\hat{\mathbf{f}}_0^+$ and $\mathbf{P}_0^+$. 
\FOR{$t=1,3,5,\cdots$}
\STATE Update $\hat{\mathbf{f}}_t^-$, $\mathbf{P}_t^-$ by (\ref{prior estimation}) and (\ref{a prior cov});
\STATE Solve optimization problem (\ref{horizontwo}) to get $\tilde{\mathbf{D}}_t$ and $\tilde{\mathbf{D}}_{t+1}$;
\STATE Calculate the sampling budget of time step $t$ by $M^*_t=\text{round}(\tr(\tilde{\mathbf{D}}_t))$ and the sampling budget of time step $t+1$ by $M^*_{t+1} = 2M-M^*_t$;
\STATE Sampling the $M^*_t$ and $M^*_{t+1}$ vertices with largest $d_{i,t}$ and $d_{i,t+1}$ in time step $t$ and $t+1$, respectively;
\STATE Update $\hat{\mathbf{f}}_t^+$ and $\mathbf{P}_t^+$ by (\ref{posterior estimation}) and (\ref{a posteriori cov}).
\ENDFOR
\end{algorithmic}
\end{algorithm}
%\vspace{-0.3cm}

\section{Analysis}
\label{section:analysis}
The optimal relaxed solution for $(\tilde{\mathbf{D}}_t,\tilde{\mathbf{D}}_{t+1})$ in (\ref{horizontwo}) can be solved by any standard optimization tool if it is convex.
So in this section, we are going to analyze the convexity of object function (\ref{horizontwo}).
%\vspace{-0.1cm}
\subsection{Convexity Composition for Matrix Valued Functions}
Object function (\ref{horizontwo}) is a composite function of $(\tilde{\mathbf{D}}_t,\tilde{\mathbf{D}}_{t+1})$.
Usually, the convexity of composition function is analyzed by the second derivative as in \cite[Sec. 3.2.4]{boyd2004convex}.
But it is hard to calculate the derivative for the matrix valued function, so we propose a new tool to analyze the convexity.

For a symmetric matrix $\mathbf{X}$, $\mathbf{X}\succeq (\preceq)\, 0$ means it is positive (negative) semidefinite. For two positive or negative semidefinite matrices $\mathbf{X}_1,\mathbf{X}_2$, $\mathbf{X}_1\succeq(\preceq)\mathbf{X}_2$ means matrix $\mathbf{X}_1-\mathbf{X}_2$ is positive or negative semidefinite.
Suppose $f:\mathbf{S}^n_{+(-)}\mapsto \mathbf{S}^m_{+(-)}$ is a matrix valued function, where $\mathbf{S}^n_{+(-)}$ denotes the
set of symmetric positive (negative) semidefinite $n\times n$ matrices.
\begin{definition}
Function $f$ is matrix nonincreasing (nondeacreasing) if $f(\mathbf{X}_1) \preceq(\succeq) f(\mathbf{X}_2)$ for $\mathbf{X}_1\succeq\mathbf{X}_2$. 
\end{definition}
\begin{definition}\cite[Sec. 3.6.2]{boyd2004convex}
Function $f$ is matrix convex (concave) with respect to matrix inequality if
\begin{equation}
\label{matrix convexity}
%\vspace{-0.1cm}
f(\theta \mathbf{X}_1 + (1-\theta)\mathbf{X}_2) \preceq (\succeq) \theta f(\mathbf{X}_1) + (1-\theta)f(\mathbf{X}_2)\nonumber
\end{equation}
for $\mathbf{X}_1,\mathbf{X}_2 \in \mathbf{S}^n_+$ or $\mathbf{X}_1,\mathbf{X}_2 \in \mathbf{S}^n_-$ and any $\theta\in [0,1]$.
\end{definition}
%An equivalent definition is that the scalar function $\mathbf{z}^\text{T}f(\mathbf{X})\mathbf{z}$ is convex for all vectors $\mathbf{z}$.
\begin{theorem}[Rule 1]
\label{theorem1}
Let $h_s:\mathbf{S}^m_{+(-)}\mapsto \mathbf{S}^k_{+(-)}$ and $h_m:\mathbf{S}^n_{+(-)}\mapsto\mathbf{S}^m_{+(-)}$, $h = h_s\circ h_m$ is matrix convex if $h_s$ is matrix convex and nonincreasing and $h_m$ is matrix concave.
\end{theorem}
\begin{proof}
We can obtain the following inequalities
%\vspace{-0.15cm}
\begin{align}
\label{concavity}
&h_s(h_m(\theta \mathbf{X}_1 + (1-\theta )\mathbf{X}_2)) \nonumber\\
\preceq\, & h_s(\theta h_m(\mathbf{X}_1) + (1-\theta)h_m(\mathbf{X}_2))\nonumber\\
\preceq\, &\theta h_s(h_m(\mathbf{X}_1))+(1-\theta)h_s(h_m(\mathbf{X}_2)),
\end{align}
where the first inequality comes from the matrix concavity of $h_m$ and matrix nonincreasing property of $h_s$, and the second inequality comes from the matrix convexity of $h_s$.
Thus, the theorem is proved.
\end{proof}
When $k=1$, $h_s$ will be a scalar valued function and  the '$\preceq$' in the second line of (\ref{concavity}) will become '$\leq$'.
We can also get the other three composition rules using the similar method as follow:

\begin{itemize}
\item Rule 2: $h = h_s\circ h_m$ is matrix convex if $h_s$ is matrix convex and nondecreasing and $h_m$ is matrix convex.
\item Rule 3: $h = h_s\circ h_m$ is matrix concave if $h_s$ is matrix concave and nonincreasing and $h_m$ is matrix convex.
\item Rule 4: $h = h_s\circ h_m$ is matrix concave if $h_s$ is matrix concave and nondecreasing and $h_m$ is matrix concave.
\end{itemize}
These composition rules can be applied multiple times to analyze the convexity or concavity of matrix valued functions.
\subsection{Convexity of Object Function}
By using Theorem \ref{theorem1}, we can get the following lemmas.
\begin{lemma}
\label{lemma2}
$h_1(\mathbf{X}) = \tr(\mathbf{X}^{-1})$ is convex and nonincreasing for $\mathbf{X}\subset\mathbf{S}^n_{++}$. $\tilde{h}_1(\mathbf{X})=-\tr(\mathbf{X}^{-1})$ is also convex and nonincreasing for $\mathbf{X}\subset\mathbf{S}^n_{--}$, where $\mathbf{S}^n_{++(--)}$ denotes a set of symmetric positive (negative) definite matrices.
\end{lemma}
%Lemma \ref{lemma2} can be proved by considering an arbitrary line \cite[Sec. 3.1.5]{boyd2004convex}, given by $\mathbf{X} = \mathbf{Z}+t\mathbf{Y}$, where $\mathbf{Z},\mathbf{Y}\in \mathbf{S}^n$ and $\mathbf{Z}+t\mathbf{Y} \succ 0$ and prove that $h_1(\mathbf{Z}+t\mathbf{Y}) = \tilde{h}_1(t)$ is convex and nonincreasing for scalar $t$.
\begin{lemma}
%\vspace{-0.15cm}
\label{lemma3}
$h_2(\mathbf{X}) = -\mathbf{A}^\text{T}\mathbf{X}^{-1}\mathbf{A}-\mathbf{B}$ is concave and nondecreasing for $\mathbf{X,A,B}\in\mathbf{S}^n_{++}$. $\tilde{h}_2(\mathbf{X})=\mathbf{A}^\text{T}\mathbf{X}^{-1}\mathbf{A}+\mathbf{B}$ is also concave and nondecreasing for $\mathbf{X,A,B}\in\mathbf{S}^n_{--}$.
\end{lemma}
%\vspace{-0.15cm}
%We can also prove Lemma \ref{lemma3} by considering an arbitrary line and define $\tilde{h}_2(t) = h_2(\mathbf{Z}+t\mathbf{Y})$.
%Then prove that $\mathbf{z}^T\tilde{h}_2(t)\mathbf{z}$ is concave and nondecreasing for all vectors $\mathbf{z}$.
%\vspace{-0.15cm}
\begin{theorem}
\label{theorem2}
The object function in (\ref{horizontwo}) is a convex function of the relaxed $(\tilde{\mathbf{D}}_t,\tilde{\mathbf{D}}_{t+1})$.
\end{theorem}
\begin{proof}
See Appendix A.
\end{proof}
\section{Numerical Results}
\label{sec:numerical evaluation}
We now numerically evaluate the performance of the proposed work. The experiments compare the proposed work with the following three methods: M1 \cite{shen2014sensor}, M2 \cite{isufi2020observing} and random sampling by normalized MSE (NMSE) 
\begin{align*}
\text{NMSE}(t) = \frac{\Vert\hat{\mathbf{f}}_t^+ - \hat{\mathbf{f}}_t\Vert^2_2}{\Vert \hat{\mathbf{f}}_t\Vert^2_2}.
\end{align*}
M1 only considers the generalized information gain which is only related to observation model but ignores the signal evolution. M2 takes signal evolution into consideration but ignores the long-term performance.

We first simulate the process of a heat source moving in a sensor network.
The sensor network is modelled as a graph with 100 vertices randomly placed in a unit square and the edges exist between vertices of which the distance is no more than 0.6. 
The heat source moves in a given trajectory which is generated by a random walk.
The evolution matrix of GS is given by the a graph translation  \cite{shuman2013emerging} operator according to the trajectory.
For example, if the center vertices of the trajectory for two continuous time steps are vertex $a$ and $b$, the evolution matrices will be $\tilde{H}_t=\text{diag}(\mathbf{V}^\text{T}\bm{\delta}_a)$ and $\tilde{H}_{t+1}=\text{diag}(\mathbf{V}^\text{T}\bm{\delta}_b)$.
The energy of the GS at each time step is normalized to 1.
The evolution and observation noise are $\textit{i.i.d}$ zero-mean Gaussian white noise with $\sigma_{\mathbf{v}}^2 = 10^{-4}$ and $\sigma_{\mathbf{w}}^2 = 10^{-3}$.  
%The discount factor $\gamma$ is set to 0.98. 
The initialization states of the GS are $\hat{\mathbf{f}}_0^+ = \mathbf{1}_{N\times 1}$ and $\mathbf{P}_0^+ = \mathbf{I}_{N\times N}$.
The discount factor $\gamma$ is set to $\gamma=0.8$.
The average sampling budget $M = 10$ and the largest budget of each time $M_t = 20$.
For the compared algorithms, the sampling budget of each time step is fixed to 10. 
%We consider the signal evolution for 1000 time steps, we improved the tracking performance by $95.12\%, 15.59\%$ and $85.24\%$ compared to M1, M2 and M3, respectively.
The accumulated tracking error for 1000 time steps is shown in the second line of Table. \ref{table:error}.
The step-by-step tracking performance of the first 100 time steps is shown in Fig. \ref{error}.

We can find that M1 and random sampling almost lose tracking of the GS and the proposed algorithm improves the tracking performance significantly compared to M2 when the signal evolution between two time steps is abrupt.
A visualized demonstration of the abrupt heat source translation from time step 31 to 32 is shown in Fig. \ref{GS} (a) (b), the circled vertices are the sampled vertices among which the center vertex is in the full line circle and the others are in dashed line circles.
It can be seen that since the long-term performance is considered, the proposed algorithm allocates more sampling budget to time step 32 to make the estimation at step 32 more accuracy.

%\begin{table}[!t]
%\small
%\caption{Accumulated tracking error of different algorithms.}
%\label{table:error}
%\centering
%\begin{tabular}{c|c|c|c|c|c}
%\hline
%\multicolumn{3}{c|}{Proposed}&\multicolumn{3}{c}{Others}\\
%\hline
%  $\gamma=0.98$ & $\gamma=0.7$ & $\gamma=0.5$& M1 \cite{shen2014sensor} & M2 \cite{isufi2017observing} & Random \\
%\hline
%22.438 & 22.412 &  23.159 & 460.124 & 26.582 & 151.973\\
%\hline
%\end{tabular}
%\end{table}

\begin{figure}[!t]
%\begin{minipage}[b]{0.9\linewidth}
  \centering
  \centerline{\includegraphics[width=9.5cm]{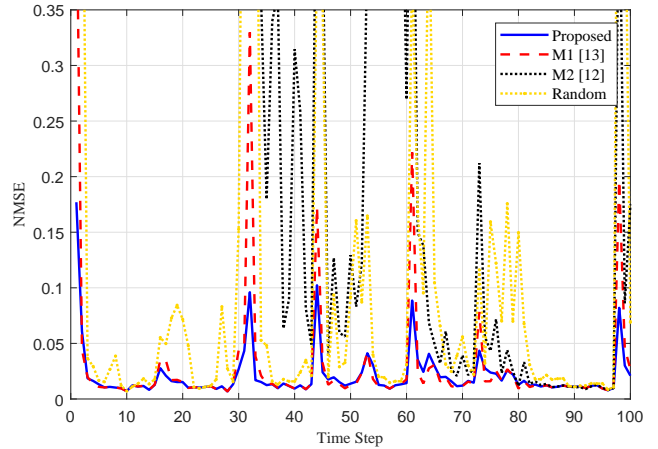}}
%\end{minipage}
%\vspace{-0.3cm}
\caption{Tracking performance of different algorithms on sensor network.}
\label{error}
\end{figure}

\begin{figure}[!t]
\centering
\begin{minipage}[b]{.4\linewidth}
  \centering
  \centerline{\includegraphics[width=4cm]{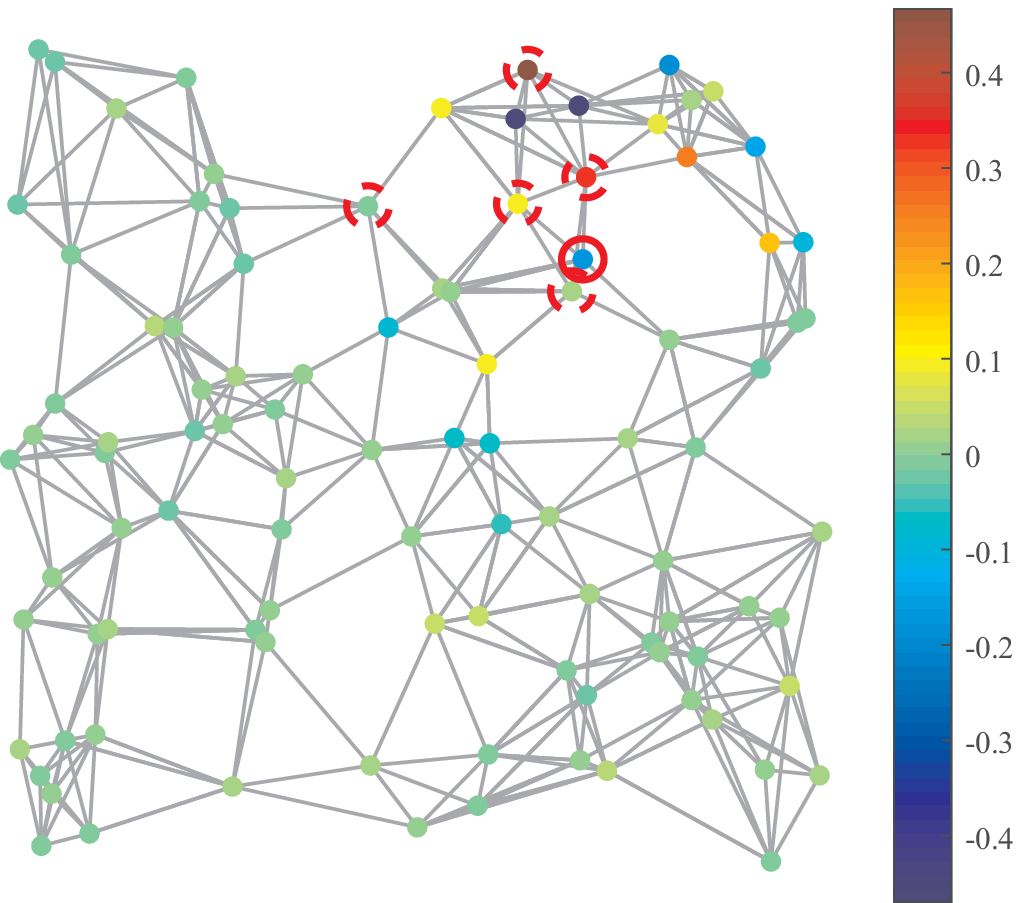}}
%  \vspace{1.5cm}
  \centerline{(a)}\medskip
\end{minipage}
\quad \quad
\begin{minipage}[b]{0.4\linewidth}
  \centering
  \centerline{\includegraphics[width=4cm]{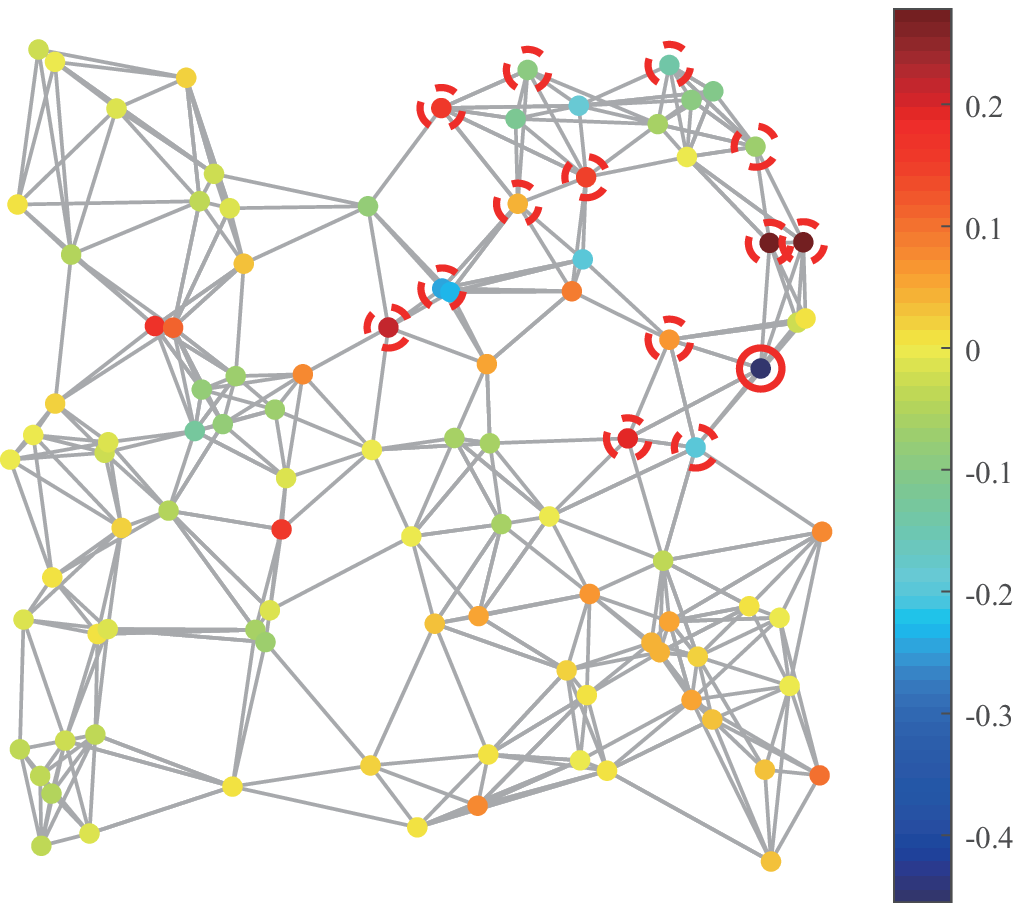}}
%  \vspace{1.5cm}
  \centerline{(b)}\medskip
\end{minipage}
%\vspace{-0.5cm}
\caption{GS and sampling sets on sensor network at time step 31 (a) and 32 (b).}
\label{GS}
\end{figure}

Next, we show that our sampling policy also fits for the tracking of GS on time-varying topology, like RES graph.
The opinion evolution in the social network is taken as an example.
A community graph with 7 communities is used to model a social network.
The probability that a edge $e= (i, j)$ in the edge set $\mathcal{E}$ is activated at time $t$ is set to $p_{i,j}=0.5$. 
The edges are activated independently across time.
At each time step $t$, we draw a graph realization $\mathcal{G}_t = (\mathcal{V},\mathcal{E}_t)$ from the underlying graph $\mathcal{G} = (\mathcal{V}, \mathcal{E})$, where the edge set $\mathcal{E}_t \subseteq  \mathcal{E}$ is generated via an i.i.d. Bernoulli process.
% and the corresponding adjacency matrix denotes $\mathbf{W}_t$.
An example of the RES community graph is shown in Fig. \ref{RES}.
\begin{figure}[!t]
\centering
\begin{minipage}[b]{.3\linewidth}
  \centering
  \centerline{\includegraphics[width=2.5cm]{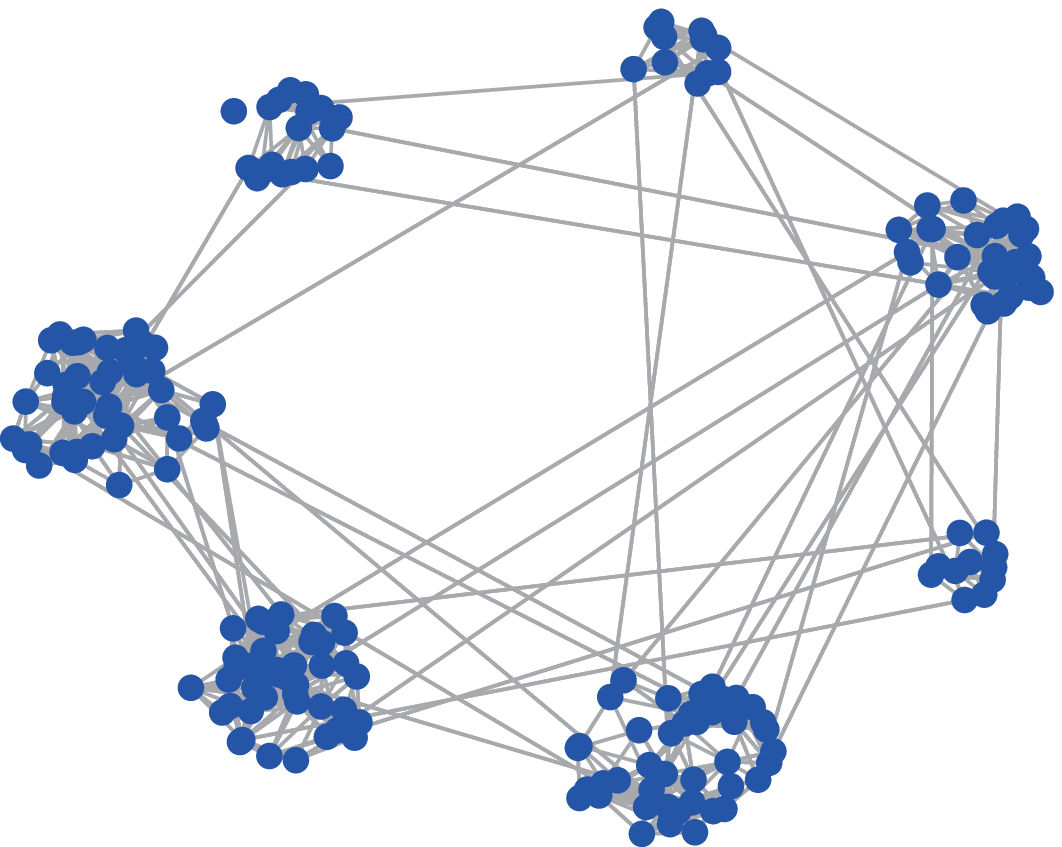}}
%  \vspace{1.5cm}
  \centerline{(a) $\mathcal{G}$}\medskip
\end{minipage}
\begin{minipage}[b]{0.3\linewidth}
  \centering
  \centerline{\includegraphics[width=2.5cm]{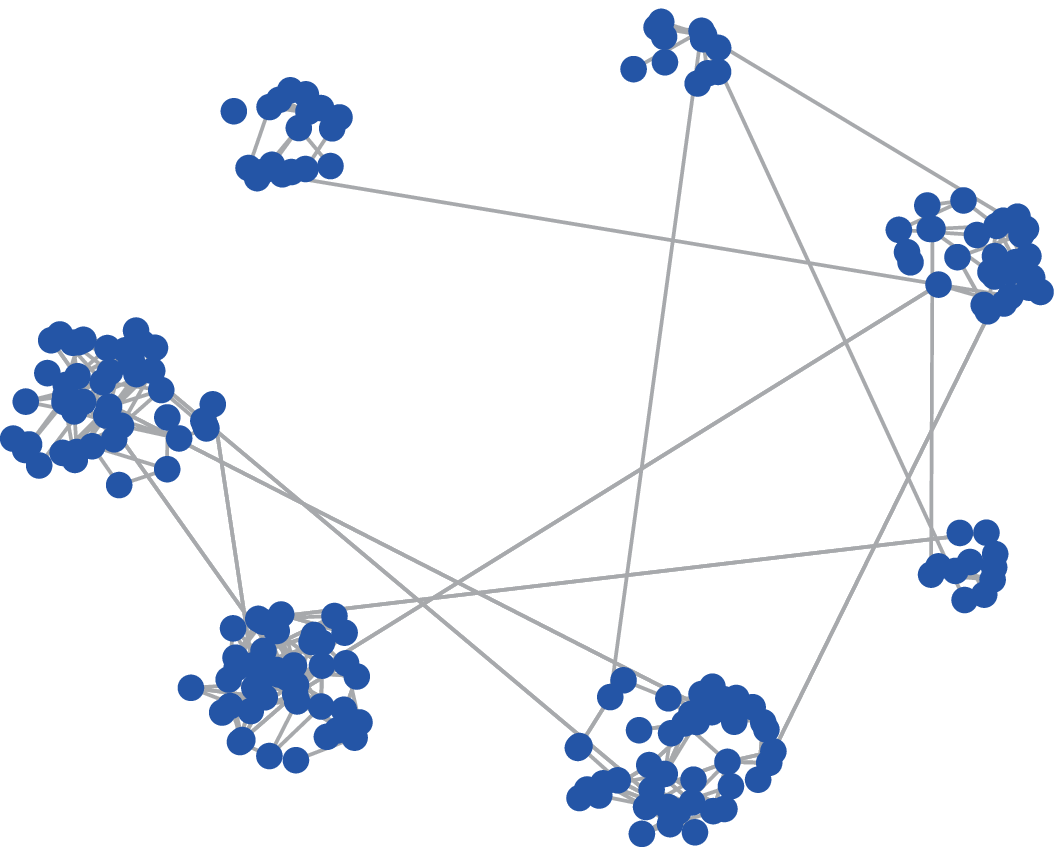}}
%  \vspace{1.5cm}
  \centerline{(b) $\mathcal{G}_1$}\medskip
\end{minipage}
\begin{minipage}[b]{0.3\linewidth}
  \centering
  \centerline{\includegraphics[width=2.5cm]{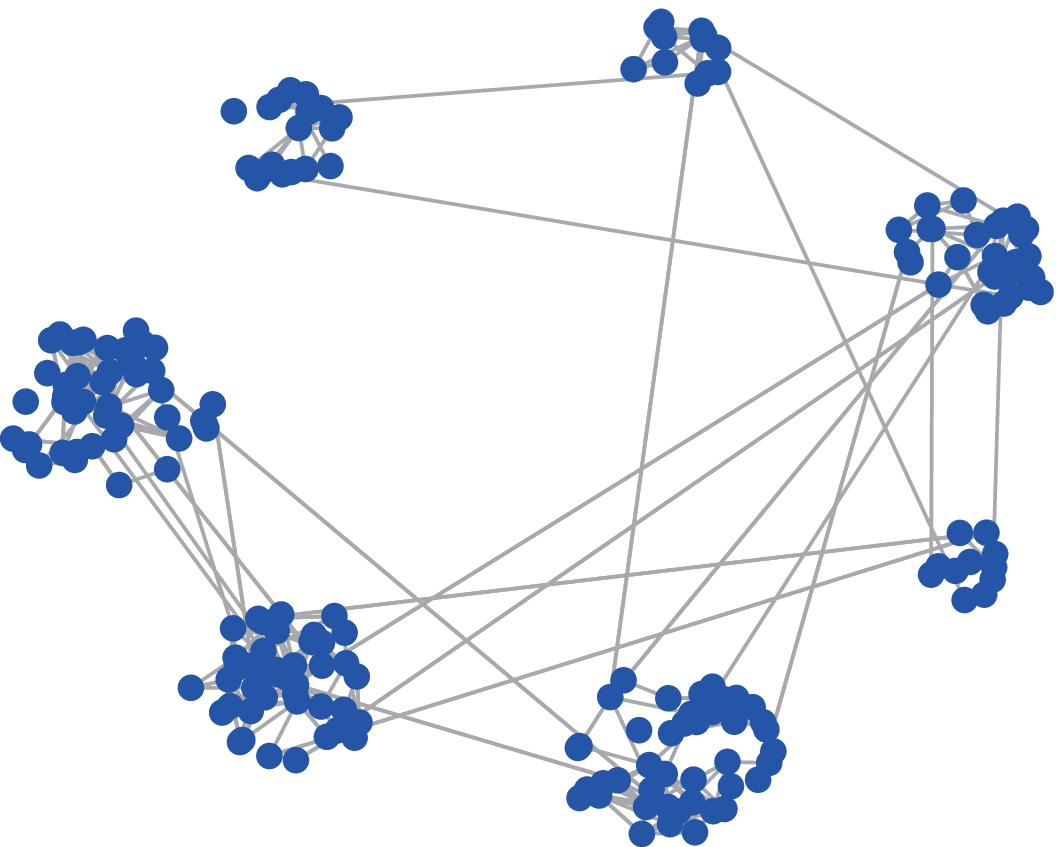}}
%  \vspace{1.5cm}
  \centerline{(c) $\mathcal{G}_5$}\medskip
\end{minipage}
%\vspace{-0.5cm}
\caption{Example of RES community graphs.}
\label{RES}
\end{figure}
\begin{table}[!t]
\caption{Accumulated Tracking Error of Different Algorithms.}
\centering
\begin{tabular}{c|c|c|c|c}
\hline
  & Proposed & M1 \cite{shen2014sensor} & M2 \cite{isufi2020observing} & Random \\
\hline
  \tabincell{c}{Sensor network\\(1000 time steps)} & 22.412 & 460.124 & 26.582 & 151.973\\
\hline  
  \tabincell{c}{Social network\\(100 time steps)} & 0.7186 & 0.8004 & 0.7385 & 0.7865\\
\hline
\end{tabular}
\label{table:error}
\end{table}
\begin{figure}[!t]
%\begin{minipage}[b]{0.9\linewidth}
  \centering
  \centerline{\includegraphics[width=9.5cm]{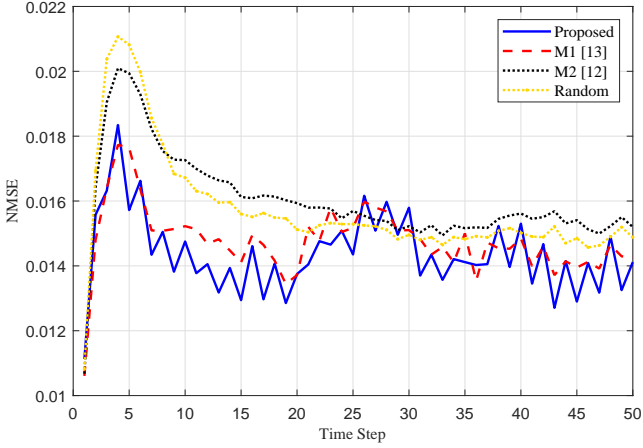}}
%\end{minipage}
%\vspace{-0.3cm}
\caption{Tracking performance of different algorithms on social network.}
\label{error2}
\end{figure}

The opinion dynamics of individuals in the network follows the Krause-Hegselmann's model \cite{Hegselmann2007Opinion}, which consider the opinion evolution of the individuals as a weighted average of their opinions at a previous time with bounded confidence.
The evolution of the GS on vertex $i$ at time step $t$ follows
\begin{align}
(\mathbf{f}_t)_i = |I(i,\mathbf{f}_{t-1})|^{-1}\sum_{j\in I(i,\mathbf{f}_{t-1})}(\mathbf{f}_{t-1})_j,
\label{KHmodel}
\end{align}
where $I(i,\mathbf{f}_{t}) = \left\{1\leq j \leq N|\, \left|(\mathbf{f}_{t})_i-(\mathbf{f}_{t})_j\right|\leq \epsilon\right\}$ and $|I(i,\mathbf{f}_{t})|$ denotes the cardinality of $I(i,\mathbf{f}_{t})$.
The opinions of individuals are initialized by uniform random numbers in $[0,1]$ with $\mathbf{P}_0^+ = 0.1\mathbf{I}_{N\times N}$ and $\epsilon$ is set to 0.3.
The energy of the GS at each time step is also normalized to 1.
The evolution and observation noise are $\textit{i.i.d.}$ zero-mean Gaussion white noise with $\sigma_{\mathbf{v}}^2 = 10^{-4}$ and $\sigma_{\mathbf{w}}^2 = 10^{-4}$.  
The average sampling budget is set to $M=10$ and the largest budget of each time step is $M_t = 20$, and the discount factor is set to $\gamma=0.8$.
The accumulated tracking error for 100 time steps is shown in the third line of Table. \ref{table:error}, and the step-by-step tracking performance is shown in Fig. \ref{error2}.
In this case, the opinions on the social network become more and more smooth with time passing by according to (\ref{KHmodel}), and therefore a more accurate estimation of GS in the former time step will help to estimate the GS in the later time step more accurately.
By allocating more samples to the former time step, algorithm \ref{algorithm:DP} also performance better in the long-term compared to the other methods with fixed sampling budget in each time step.
\section{Conclusion}
\label{sec:con}
In this paper, a sampling policy with adaptive budget allocation is proposed for tracking a time-varying graph signal with KF. 
By considering the influence of the current sampling policy to the future performance, we formulate the problem as an infinite horizon sequential decision process.
An approximate solution is obtained by truncating the future horizon to one, which improves the tracking performance a lot. 
%In the future work, we are going to design a sampling policy which consider more future horizons performance.
\vfill
\stripsep-3pt
\begin{strip}
\rule{\textwidth}{0.3mm}
\begin{align}
\label{equ1}
&\gamma \tr\left(\left(\tilde{\mathbf{H}}_t\left((\mathbf{P}_t^-)^{-1} +  \sigma_{\mathbf{w}}^{-2}\mathbf{V}^\text{T}\tilde{\mathbf{D}}_t\mathbf{V}\right)^{-1}\tilde{\mathbf{H}}_t + \mathbf{\Sigma}_{\mathbf{v}}\right)^{-1} + \sigma_{\mathbf{w}}^{-2}\mathbf{V}^\text{T}\tilde{\mathbf{D}}_{t+1}\mathbf{V}\right)^{-1}\nonumber\\
=&-\gamma \tr\left(\left(-\tilde{\mathbf{H}}_t\left((\mathbf{P}_t^-)^{-1} +  \sigma_{\mathbf{w}}^{-2}\mathbf{V}^\text{T}\tilde{\mathbf{D}}_t\mathbf{V}\right)^{-1}\tilde{\mathbf{H}}_t - \mathbf{\Sigma}_{\mathbf{v}}\right)^{-1} - \sigma_{\mathbf{w}}^{-2}\mathbf{V}^\text{T}\tilde{\mathbf{D}}_{t+1}\mathbf{V}\right)^{-1}.
\end{align}
\begin{align}
&\left(-\tilde{\mathbf{H}}_t\left((\mathbf{P}_t^-)^{-1} +  \sigma_{\mathbf{w}}^{-2}\mathbf{V}^\text{T}(\theta\mathbf{X}_1+(1-\theta)\mathbf{Y_1})\mathbf{V}\right)^{-1}\tilde{\mathbf{H}}_t - \mathbf{\Sigma}_{\mathbf{v}}\right)^{-1} - \sigma_{\mathbf{w}}^{-2}\mathbf{V}^\text{T}(\theta\mathbf{X}_2+(1-\theta)\mathbf{Y_2})\mathbf{V}\nonumber\\
\succeq\, & \theta\left(-\tilde{\mathbf{H}}_t\left((\mathbf{P}_t^-)^{-1} +  \sigma_{\mathbf{w}}^{-2}\mathbf{V}^\text{T}\mathbf{X}_1\mathbf{V}\right)^{-1}\tilde{\mathbf{H}}_t - \mathbf{\Sigma}_{\mathbf{v}}\right)^{-1} + (1-\theta)\left(-\tilde{\mathbf{H}}_t\left((\mathbf{P}_t^-)^{-1} +  \sigma_{\mathbf{w}}^{-2}\mathbf{V}^\text{T}\mathbf{Y}_1\mathbf{V}\right)^{-1}\tilde{\mathbf{H}}_t - \mathbf{\Sigma}_{\mathbf{v}}\right)^{-1}\nonumber\\
&  -\theta \sigma_{\mathbf{w}}^{-2}\mathbf{V}^\text{T}\mathbf{X}_2\mathbf{V}-(1-\theta) \sigma_{\mathbf{w}}^{-2}\mathbf{V}^\text{T}\mathbf{Y}_2\mathbf{V}.
\label{equ2}
\end{align}
\rule{\textwidth}{0.3mm}
\vspace{0.1mm}
\end{strip}
\section*{Acknowledgement}
This work was supported by the Shanghai Municipal Natural Science Foundation (19ZR1404700), National Major Scientific Research Instruments and Equipments Development Project of NSFC (11827808), and Fudan-Zhuhai Innovation Institute.
\section*{Appendix A The Proof of Theorem 2}
\label{appendix}
\textit{Theorem} 2: The object function in (\ref{horizontwo}) is a convex function of the relaxed $(\tilde{\mathbf{D}}_t,\tilde{\mathbf{D}}_{t+1})$.\\
\begin{proof}
Obviously, $\sigma_{\mathbf{w}}^{-2}\mathbf{V}^\text{T}\tilde{\mathbf{D}}_t\mathbf{V}$ is a linear function of $\tilde{\mathbf{D}}_t$.
Since $(\mathbf{P}_t^-)^{-1} +\sigma_{\mathbf{w}}^{-2}\mathbf{V}^\text{T}\tilde{\mathbf{D}}_t\mathbf{V}\in \mathbf{S}_+^N$, using Lemma 1 and composition Rule 1 we can prove that the first term of (\ref{horizontwo}) is a convex function of $(\tilde{\mathbf{D}}_t,\tilde{\mathbf{D}}_{t+1})$.

The second term of (\ref{horizontwo}) can be rewritten as (\ref{equ1}). 
For easier reading, let
\begin{equation}
\label{Z1}
\mathbf{Z}_1 = -\tilde{\mathbf{H}}_t\left((\mathbf{P}_t^-)^{-1} +  \sigma_{\mathbf{w}}^{-2}\mathbf{V}^\text{T}\tilde{\mathbf{D}}_t\mathbf{V}\right)^{-1}\tilde{\mathbf{H}}_t - \mathbf{\Sigma}_{\mathbf{v}}.
\end{equation}
Using Lemma 2 and composition Rule 4, we can prove that $\mathbf{Z}_1$ is a concave function of $\tilde{\mathbf{D}}_t$.

It is obvious that $\mathbf{Z}_1$ is symmetric negative semidefinite. 
So we can prove that $\mathbf{Z}_1^{-1}$ is a concave function of $\tilde{\mathbf{D}}_t$ using Lemma 2 and composition Rule 4.

Since $\sigma_{\mathbf{w}}^{-2}\mathbf{V}^\text{T}\tilde{\mathbf{D}}_{t+1}\mathbf{V}$ is a linear function of $\tilde{\mathbf{D}}_{t+1}$, for \\$\mathbf{X}_1,\mathbf{Y}_1,\mathbf{X}_2,\mathbf{Y}_2 \subset \tilde{\mathcal{A}}$ and $\theta\in [0,1]$, we have the (\ref{equ2}).
According to (\ref{Z1}), let
\begin{equation}
\label{Z2}
\mathbf{Z}_2 = \mathbf{Z}_1^{-1} - \sigma_{\mathbf{w}}^{-2}\mathbf{V}^\text{T}\tilde{\mathbf{D}}_{t+1}\mathbf{V}.
\end{equation}
Thus, $\mathbf{Z}_2$ is a concave function of $(\tilde{\mathbf{D}}_t,\tilde{\mathbf{D}}_{t+1})$.

Since $\mathbf{Z}_2\in \mathbf{S}_{--}^N$, we can prove that the second term of (\ref{horizontwo}) is a concave function of $(\tilde{\mathbf{D}}_t,\tilde{\mathbf{D}}_{t+1})$ using Lemma 1 and composition Rule 4.
%
%Thus, Theorem 2 is proved.
\end{proof}
%\vfill\pagebreak

\bibliographystyle{IEEEtran}
\bibliography{refs}
\end{document}